\documentclass[letterpaper, 10 pt, conference]{ieeeconf}  

\IEEEoverridecommandlockouts             
\usepackage{graphicx,amsmath,amssymb}
\usepackage{kotex}
\usepackage{algorithm}
\usepackage{algorithmic}
\usepackage{tikz}
\usetikzlibrary{arrows.meta, positioning, quotes, calc, shapes.geometric}

\tikzset{
	block/.style = {draw, rectangle,
		minimum height=1cm,
		minimum width=2cm},
	input/.style = {coordinate,node distance=1cm},
	output/.style = {coordinate,node distance=4cm},
	arrow/.style={draw, -latex,node distance=2cm},
	pinstyle/.style = {pin edge={latex-, black,node distance=2cm}},
	sum/.style = {draw, circle, node distance=1cm},
}

\usepackage{amsfonts}
\usepackage{amsmath}
\usepackage{amssymb}

\usepackage{amsthm}

\newtheorem{thm}{\bf Theorem}

\newtheorem{lem}{\bf Lemma}
\newtheorem{cor}{\bf Corollary}

\theoremstyle{definition}
\newtheorem{defn}{\bf Definition}

\newtheorem{quest}{\bf Question}

\theoremstyle{remark}
\newtheorem{rem}{\bf Remark}

\title{\LARGE \bf Steady-state response assignment for a given disturbance and reference: Sylvester equation rather than regulator equations}

\author{Hyeonyeong Jang and Jin Gyu Lee
\thanks{H.~Jang and J.G.~Lee are with ASRI and the Department of Electrical and Computer Engineering, Seoul National University, Korea (email: {\tt\small hyjang@cdsl.kr} and {\tt\small jingyu.lee@snu.ac.kr}). Corresponding author: Jin Gyu Lee.}
}

\begin{document}

\maketitle
\thispagestyle{plain} 
\pagestyle{plain} 

\begin{abstract}
Conventionally, the concept of moment has been primarily employed in model order reduction to approximate system by matching the moment, which is merely the specific set of steady-state responses.
In this paper, we propose a novel design framework that extends this concept from ``moment matching'' for approximation to ``moment assignment'' for the active control of steady-state.
The key observation is that the closed-loop moment of an interconnected linear system can be decomposed into the open-loop moment and a term linearly parameterized by the moment of the compensator.
Based on this observation, we provide necessary and sufficient conditions for the assignability of desired moment and a canonical form of the dynamic compensator, followed by constructive synthesis procedure of compensator.
This covers both output regulation and closed-loop interpolation, and further suggests using only the Sylvester equation, rather than regulator equations.
\end{abstract}


\section{Introduction}

The time domain notion of moment characterizes the input-output behavior of interconnected system on a specific invariant manifold~\cite{astolfi2010model}.
In particular, consider the multi-input-multi-output (MIMO) linear time-invariant (LTI) system described by
\begin{equation*}
\begin{aligned}
\dot x & = Ax + Bu, &
y &= Cx + Du,
\end{aligned}
\end{equation*}
with the state $x\in\mathbb{R}^n$, the input $u\in\mathbb{R}^{m}$, and the output $y\in\mathbb{R}^{p}$;
and the signal generator
\begin{equation*}
\begin{aligned}
\dot \omega & = S\omega, & v &= L\omega,
\end{aligned}
\end{equation*}
with the state $\omega \in \mathbb{R}^\nu$ and the output $v\in\mathbb{R}^{m}$.
The injection of $v$ into $u$ yields\ open-loop interconnected system
\begin{equation*}
\begin{aligned}
\dot \omega &= S\omega \\
\dot x &= Ax + BL\omega\\
y &= Cx + DL\omega.
\end{aligned}
\end{equation*}
If the spectra of A and S are disjoint, then the interconnected system possesses the invariant manifold
\begin{equation*}
\Upsilon_{x\omega} = \{ (x,\omega) \in \mathbb{R}^{n + \nu} \mid x = \Pi \omega \}
\end{equation*}
where $\Pi$ is the unique solution of the Sylvester equation
\begin{equation*}
\Pi S = A\Pi + BL.
\end{equation*}
Indeed, on this invariant manifold, the input($\omega$)-output($y$) behavior $y = (C\Pi + DL)\omega$ can be abstracted into a single linear operator $C\Pi + DL$, which we call \emph{moment}.
If, in addition, $A$ is Hurwitz, then it is straightforward from the error dynamics of $e = x - \Pi \omega$ that the invariant manifold becomes attractive, hence the moment can fully describe the steady-state response.

Under this observation, a series of research focusing on reducing the model while matching the moment, hence preserving the exact steady-state response for a particular set of excitation of interest has been conducted for various classes of systems including linear, nonlinear, hybrid, and stochastic systems~\cite{astolfi2010model,scarciotti2016model,scarciotti2021moment}.

Meanwhile, if an additional feedback loop is applied to the system, the invariant manifold, in general, shifts, consequently altering the moment of the overall interconnected system.
For the model reduction by moment matching, this alteration is often regarded as a malfunction to be avoided.
In particular, recent efforts termed as closed-loop interpolation have focused on characterizing a specific class of feedback controllers that strictly preserve the original moment of the system, while providing attractivity~\cite{moreschini2024closed}.

In this paper, rather than viewing the feedback-induced alteration of moments as an unintended consequence, we exploit this phenomenon as a control design tool.
Since the moment fully dictates the steady-state response, by intentionally shifting the moment via feedback we can straightforwardly assign a predefined steady-state response for a given disturbance/reference, with the only additional effort of making the invariant manifold attractive.
In particular, via this preliminary work, we suggest the moment-based framework to transit from passive ``moment matching'' for modeling to active ``moment assignment'' for disturbance rejection and tracking control.
By doing so, this work formally bridges the gap between passive moment matching and active output regulation.
Furthermore, this moment-based abstraction drastically simplifies the control design, yielding the exact necessary and sufficient conditions for steady-state assignment and allowing for the complete parameterization of all moment-assigning controllers in a canonical form, hence suggesting the use of the Sylvester equation rather than regulator equations.

The paper is organized as follows.
Section II formulates the problem of interest.
Section III analyzes the modified moment via the plant-compensator interconnection (III-A); derives necessary and sufficient conditions for moment assignability (III-B); and formally characterizes this condition introducing the moment transfer operator (III-C).
Section IV parameterizes all moment-assigning compensators (IV-A) and provides necessary and sufficient conditions to stabilize the closed-loop system, followed by a constructive synthesis procedure (IV-B).
Sections V and VI contain a numerical example and concluding remarks.

\emph{Notation:} Let $\mathbb{C}$, $\mathbb{R}$, and $\mathbb{Z}$ denote the sets of complex numbers, real numbers, and integers, respectively. The real part of the complex number $z$ is denoted by $\mathfrak{R}(z)$. For a set $E \subseteq \mathbb{C}$ and a constant $a \in \mathbb{R}$, we define $E_{\ge a} := \{z \in E \mid \mathfrak{R}(z) \ge a\}$.
The $n\times n$ identity matrix is denoted by $I_n$, the $n\times n$ zero matrix is denoted by $0_n$, and the $m\times p$ zero matrix is denoted by $0_{m\times p}$.
The subscript could be omitted if it is trivial.
For the square matrix $A$, $\sigma(A)$ denotes its spectrum. For the linear operator $L$, $\mathcal{R}(L)$ and $\mathcal{N}(L)$ denote the range and the null space of $L$.
The Kronecker product is denoted by $\otimes$.
For the matrix $M\in\mathbb{R}^{p\times m}$, $\operatorname{vec}(M)\in\mathbb{R}^{pm}$ denotes the vectorization of $M$. For the matrices $M_1,\dots,M_N$, $\operatorname{blkdiag}(M_1,\dots,M_N)$ denotes the block diagonal matrix with diagonal blocks $M_1, \dots, M_n$.
Given square matrices $A\in\mathbb{R}^{n\times n}$ and $S\in\mathbb{R}^{\nu\times\nu}$, the Sylvester operator $\mathcal{L}_{(S,A)}:\mathbb{R}^{n\times\nu} \to \mathbb{R}^{n\times\nu}$ is defined by the mapping $\Pi \mapsto \Pi S - A\Pi$.

\section{Problem formulation}
\subsection{Open-loop moment}
Consider the MIMO LTI system of the equation
\begin{equation}\label{eq: LTI system with u and mu}
\begin{aligned}
\dot x &= Ax + Bu + P\mu, & y = Cx + Du + Q\mu,
\end{aligned}
\end{equation}
with the state $x\in\mathbb{R}^n$, the input $u\in\mathbb{R}^{m}$, the exogenous signal $\mu \in\mathbb{R}^q$, and the output $y\in\mathbb{R}^{p}$.
Suppose that the signal generator
\begin{equation}\label{eq: Signal generator for mu}
\begin{aligned}
\dot \omega & = S\omega, & v &= L\omega,
\end{aligned}
\end{equation}
with the state $\omega \in \mathbb{R}^\nu$ and the output $v\in\mathbb{R}^q$,
injects its output into the system~\eqref{eq: LTI system with u and mu} by $\mu =v$.

With no further input to \eqref{eq: LTI system with u and mu}, i.e.,  $u \equiv 0$, following the convention illustrated in the introduction we call $M_\mathrm{open}:=C\Pi + QL$ as the \emph{open-loop moment} of the system \eqref{eq: LTI system with u and mu} at $(S,L)$ where $\Pi = \mathcal{L}_{(S,A)}^{-1}(PL)$ is the unique solution of the Sylvester equation
\begin{equation*}
\Pi S = A\Pi + PL.
\end{equation*}

For this preliminary work, for the simplicity of discussion, we will have the standing assumption $\sigma(A) \cap \sigma(S) = \emptyset$, which is a necessary and sufficient condition for the invertibility of the Sylvester operator~\cite{bellman1997introduction}.

\subsection{Dynamic compensator \& closed-loop moment}
Consider the dynamic compensator of the general form
\begin{equation}\label{eq: compensator}
\begin{aligned}
\dot \xi &= F \xi + G u_\xi, & y_\xi = H\xi,
\end{aligned}
\end{equation}
with the state $\xi \in \mathbb{R}^\rho$, the input $u_\xi \in \mathbb{R}^{p}$, and the output $y_\xi\in\mathbb{R}^{m}$.
Here, the matrices $F$, $G$, $H$, and the order $\rho$ are the parameters to be designed, while the dimensions of the input and output $p$, $m$ are fixed as the dimensions of the output and input of the system \eqref{eq: LTI system with u and mu}, respectively.
If \eqref{eq: compensator} is interconnected with the system~\eqref{eq: LTI system with u and mu} by $u = y_\xi$ and $u_\xi = y$, the closed-loop system becomes
\begin{equation}\label{eq: closed-loop system}
\begin{aligned}
\dot z &= \underbrace{\begin{bmatrix}
A & BH \\ GC & F + GDH
\end{bmatrix}}_{=: A_\mathrm{cl}} z + \underbrace{\begin{bmatrix}
P \\ GQ
\end{bmatrix}}_{=: P_\mathrm{cl}} \mu\\
y &= \underbrace{\begin{bmatrix}
C & DH
\end{bmatrix}}_{=:C_\mathrm{cl}} z + Q\mu
\end{aligned}
\end{equation}
with $z = [x^\top \ \xi^\top]^\top \in\mathbb{R}^{n+\rho}$.

Hypothetically suppose that the spectra of $A_\mathrm{cl}$ and $S$ are disjoint.{\footnote{Since the standing assumption $\sigma(A) \cap \sigma(S) = \emptyset$ guarantees that the uncontrollable and unobservable modes of the plant do not overlap with $S$, the disjoint spectrum condition $\sigma(A_\mathrm{cl}) \cap \sigma(S) = \emptyset$ is not restricted by these fixed modes.}}
Then the moment of the closed-loop system~\eqref{eq: closed-loop system} at $(S,L)$ is $C_\mathrm{cl}\mathcal{L}_{(S, A_\mathrm{cl})}^{-1}(P_\mathrm{cl}L) + QL = C\Pi_x + DH\Pi_\xi + QL$ where $\Pi_x\in\mathbb{R}^{n\times\nu}$ and $\Pi_\xi \in \mathbb{R}^{\rho\times\nu}$ are the unique solution of the Sylvester equation
\begin{equation}\label{eq: Sylvester equation of the closed-loop system}
\begin{bmatrix}
\Pi_x \\ \Pi_\xi
\end{bmatrix} S
=
A_\mathrm{cl}
\begin{bmatrix}
\Pi_x \\ \Pi_\xi
\end{bmatrix} + 
P_\mathrm{cl}L.
\end{equation}
We will define $M_\mathrm{cl}(F, G, H; \rho) := C\Pi_x + DH\Pi_\xi + QL$ as the \emph{closed-loop moment of the system}~\eqref{eq: LTI system with u and mu} with the compensator~\eqref{eq: compensator} at $(S,L)$.
We will also denote $M_\mathrm{c} := H\Pi_\xi$ the \emph{closed-loop moment of the compensator}~\eqref{eq: compensator} with the system~\eqref{eq: LTI system with u and mu} at $(S,L)$, which is the moment corresponding to the output $y_\xi = [0 \ \ H]z$, i.e., $M_\mathrm{c} = [0 \ \  H]\mathcal{L}_{(S, A_\mathrm{cl})}^{-1}(P_\mathrm{cl}L)$.

\subsection{Problem statement}

Given the plant~\eqref{eq: LTI system with u and mu} and the signal generator~\eqref{eq: Signal generator for mu} satisfying the assumption $\sigma(A) \cap \sigma(S) = \emptyset$, let $M_{\mathrm{des}} \in \mathbb{R}^{{p} \times \nu}$ be the desired closed-loop moment. The objective is to answer the following set of questions.

\begin{quest}
Exactly when does a dynamic compensator of the form \eqref{eq: compensator} exist such that it assigns the closed-loop moment of \eqref{eq: LTI system with u and mu} at $(S, L)$ to $M_\mathrm{des}$, i.e., there exists $(F, G, H; \rho)$ such that $M_\mathrm{cl}(F, G, H; \rho) = M_\mathrm{des}$?
\end{quest}

\begin{quest}
For a given assignable moment $M_\mathrm{des}$, would there be a canonical structure that can effectively parameterize all the compensator that assigns $M_\mathrm{des}$?
\end{quest}

\begin{quest}
In addition to the moment assignability, exactly when can the invariant manifold be attractive, i.e., there exists $(F, G, H; \rho)$ such that $M_\mathrm{cl}(F, G, H; \rho) = M_\mathrm{des}$ and $A_\mathrm{cl}$ is Hurwitz?
\end{quest}

\begin{quest}
For a given assignable steady-state response $M_\mathrm{des}$, would there be a constructive way to synthesize actual dynamic compensator that achieve the moment assignment while also ensuring the attractivity of the invariant manifold?
\end{quest}

\section{Moment assignability}
\subsection{Parameterization of closed-loop moment by compensator}
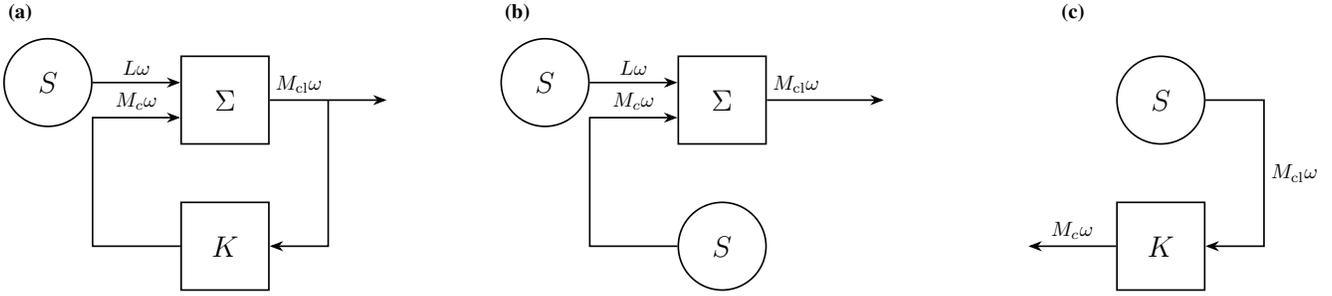
\begin{figure*}[t]
\usetikzlibrary{positioning, arrows.meta, calc}
    \centering
    \resizebox{\textwidth}{!}{
\begin{tikzpicture}[
    >=Stealth,
    block/.style={rectangle, draw, thick, minimum width=1.5cm, minimum height=1.5cm, font=\Large},
    gen/.style={circle, draw, thick, minimum size=1.5cm, font=\Large,
    },
]

\begin{scope}[shift={(0,0)}]
    \node[block] (Sigma_a) at (0,0) {$\Sigma$};
    \node[block] (K_a) at (0,-2.5) {$K$};
    
    \coordinate (in1_a) at ([yshift=0.3cm]Sigma_a.west);
    \node[gen, left=1.5cm of in1_a] (S_a) {$S$};

    \draw[->, thick] (S_a.east) -- (in1_a) node[midway, above] {$L\omega$};
    \draw[->, thick] (Sigma_a.east) -- ++(2,0) coordinate (out_a) node[near start, above] {$M_\mathrm{cl}\omega$};
    
    \draw[->, thick] ([xshift=1cm]Sigma_a.east) |- (K_a.east);
    \draw[->, thick] (K_a.west) -- ++(-1.5,0) |- ([yshift=-0.3cm]Sigma_a.west) node[near end, above] {$M_c\omega$};

    \node at (-3.5, 1.5) {\textbf{(a)}};
\end{scope}

\begin{scope}[shift={(8.5,0)}]
    \node[block] (Sigma_b) at (0,0) {$\Sigma$};
    \node[gen] (S_b2) at (0,-2.5) {$S$}; 
    
    \coordinate (in1_b) at ([yshift=0.3cm]Sigma_b.west);
    \node[gen, left=1.5cm of in1_b] (S_b1) {$S$};

    \draw[->, thick] (S_b1.east) -- (in1_b) node[midway, above] {$L\omega$};
    \draw[->, thick] (S_b2.west) -- ++(-1.5,0) |- ([yshift=-0.3cm]Sigma_b.west) node[near end, above] {$M_c\omega$};
    \draw[->, thick] (Sigma_b.east) -- ++(2,0) node[near start, above] {$M_\mathrm{cl}\omega$};

    \node at (-3.5, 1.5) {\textbf{(b)}};
\end{scope}

\begin{scope}[shift={(16,0)}]
    \node[gen] (S_c) at (0,0) {$S$}; 
    \node[block] (K_c) at (0,-2.5) {$K$};

    \draw[->, thick] (S_c.east) -- ++(1,0) coordinate (turn) |- (K_c.east);
    \node[right] at ($(turn)!0.5!(turn |- K_c.east)$) {$M_\mathrm{cl}\omega$};
    
    \draw[->, thick] (K_c.west) -- ++(-1.5,0) node[midway, above] {$M_c\omega$};

    \node at (-1.5, 1.5) {\textbf{(c)}};
\end{scope}

\end{tikzpicture}
}
    \caption{Structural decomposition of interconnected dynamics on the invariant manifold~\eqref{eq: Invariant manifold of closed-loop system}.
    (a) The original interconnected system with the signal generator of equation $\dot \omega = S \omega$ ($S$), plant ($\Sigma$), and compensator ($K$).
    (b) From the perspective of the plant $\Sigma$, incoming signals from $S$ and $K$ of (a) can be regarded as a signal from the virtual signal generator $(S, [L^\top \ \ M_\mathrm{c}^\top]^\top)$, yielding the open-loop moment $M_\mathrm{cl}$.
    Note that to assign moment to $M_\mathrm{des}$, existence of virtual a signal generator $(S,M_\mathrm{c})$ which makes $M_\mathrm{cl} = M_\mathrm{des}$ is essential.
    (c) From the  perspective of the compensator $K$, incoming signal from $\Sigma$ of (a) can be regarded as a signal from virtual signal generator $(S, M_\mathrm{cl})$, yielding the open-loop moment $M_\mathrm{c}$.
    Note that the synthesis problem can now be simplified to designing a compensator that generates the required moment $M_\mathrm{c}$ when driven by the virtual desired generator $(S, M_\mathrm{des})$.}
    \label{fig. Analysis on the invariant manifold}
\end{figure*}

In this section, we investigate the influence of the compensator on the closed-loop moment.
Partitioning the Sylvester equation~\eqref{eq: Sylvester equation of the closed-loop system} yields
\begin{subequations}\label{eq: closed sylvester - a,b}
\begin{align}
\Pi_x S &= A\Pi_x + (PL+BH\Pi_\xi) , \label{eq: closed sylvester - a}\\
\Pi_\xi S &= (F+GDH)\Pi_\xi + G(C\Pi_x +QL). \label{eq: closed sylvester - b}
\end{align}
\end{subequations}
If $\sigma(A_\mathrm{cl}) \cap \sigma(S) = \emptyset$, then there exists unique pair $(\Pi_x,\Pi_\xi)$ achieving both \eqref{eq: closed sylvester - a} and \eqref{eq: closed sylvester - b}. Furthermore, the interconnnected system of \eqref{eq: Signal generator for mu} and \eqref{eq: closed-loop system}, with $\mu=v$ possesses the invariant manifold
\begin{equation}\label{eq: Invariant manifold of closed-loop system}
\Upsilon_{x\xi\omega} = \{ (x,\xi,\omega) \in \mathbb{R}^{n + \rho + \nu} \mid x = \Pi_x \omega, \xi = \Pi_\xi \omega \},
\end{equation}
which is attractive, if and only if $A_\mathrm{cl}$ is Hurwitz.

From~\eqref{eq: closed sylvester - a}, $\Pi_x$ can be explicitly expressed as
\begin{equation} \label{eq: Pi_x decomposition}
\Pi_x = \mathcal{L}_{(S,A)}^{-1} (PL) + \mathcal{L}_{(S,A)}^{-1} (BH\Pi_\xi).
\end{equation}
Note that the first term, $\mathcal{L}_{(S,A)}^{-1} (PL)$, corresponds to the open-loop moment, and $H\Pi_\xi$ is the closed-loop moment of the compensator $M_\mathrm{c}$.
Hence, we obtain    
\begin{equation} \label{eq: M_cl expansion}\small
\begin{aligned}
&M_\mathrm{cl}(F, G, H; \rho) = C\Pi_x + DH\Pi_\xi + QL \\
&= \underbrace{\left[C\mathcal{L}_{(S,A)}^{-1} (PL) + QL\right]}_{=M_\mathrm{open}} + \underbrace{\left[C\mathcal{L}_{(S,A)}^{-1} (BM_\mathrm{c}) + DM_\mathrm{c}\right]}_{\text{moment of $(A,B,C,D)$ at $(S,M_\mathrm{c})$}}.
\end{aligned}
\end{equation}
Here, $M_\mathrm{open}$ was already introduced in Section II-A as the open-loop moment of the system~\eqref{eq: LTI system with u and mu} at $(S,L)$.
Equation~\eqref{eq: M_cl expansion} indicates that the closed-loop moment is, in fact, the superposition of the open-loop moment and the term induced by the compensator,
which can be reinterpreted as again the open-loop moment of the system identified by the quadruple $(A, B, C, D)$, but at a virtual signal generator $(S, M_\mathrm{c})$.
Together, the closed-loop moment is nothing but an open-loop moment of the system~\eqref{eq: LTI system with u and mu} at a concatenated virtual signal generator $(S, [ L^\top  \ \ M_c^\top ]^\top)$.

On the other hand, rewriting \eqref{eq: closed sylvester - b} with $M_\mathrm{cl} = C\Pi_x + DH\Pi_\xi + QL$, we get
\begin{equation}
\Pi_\xi S = F\Pi_\xi + G M_\mathrm{cl}.
\end{equation}
This again implies that the closed-loop moment of the compensator $M_\mathrm{c} = H\Pi_\xi$ can also be interpreted as the open-loop moment of the compensator~\eqref{eq: compensator} at another virtual signal generator $(S, M_\mathrm{cl})$.

In summary, the interconnected dynamics on the invariant manifold~\eqref{eq: Invariant manifold of closed-loop system} can be decoupled to its subsystems, while for each subsystem, all the incoming signals can be interpreted as generated by the virtual signal generator $(S, M_\mathrm{in})$. The output of each subsystem is again abstracted by the open-loop moment $M_\mathrm{out}$, which acts as a virtual signal generator for other subsystems.
While the exogenous system $(S, L)$ acts as the real physical signal generator, the plant and compensator operate as mutually dependent virtual signal generators.
The compensator functions as a generator $(S, M_\mathrm{c})$ where $M_\mathrm{c}$ is driven by the plant's moment $M_\mathrm{cl}$, whereas 
the plant also functions as a generator $(S, M_\mathrm{cl})$ where $M_\mathrm{cl}$ is fully described by the compensator's moment $M_\mathrm{c}$.
See also Figure~\ref{fig. Analysis on the invariant manifold}.

\subsection{Moment assignability}

According to the key observation made at the end of previous section, the specific physical realizations of the system and the compensator become insignificant for assigning the desired closed-loop moment.
Instead of considering the complex internal structure of a compensator, the synthesis problem now depends solely on the virtual signal generator it represents.
Consequently, the existence of a virtual signal generator $(S,M_\mathrm{c})$ which let $M_\mathrm{cl}$ become $M_\mathrm{des}$ is a necessary condition for the moment assignability.
In this point of view, it is crucial to know how a subsystem transfer incoming signal to moment; what would be the resulting moment $M_\mathrm{out}$ when injected by a virtual signal generator $(S,M_\mathrm{in})$.
Therefore, we first introduce the operator which explicitly maps incoming signal generated by $(S,M_\mathrm{in})$ to the open-loop moment of the given LTI system.

\begin{defn}
Given LTI system $(A,B,C,D)$ and $S\in\mathbb{R}^{\nu\times\nu}$, with $\sigma(A) \cap \sigma(S) = \emptyset$, the moment transfer operator of $(A,B,C,D)$ at $S$ is the linear mapping $\mathcal{T}_S: \mathbb{R}^{{m}\times\nu} \to \mathbb{R}^{{p}\times\nu}$ such that
\begin{equation}
\mathcal{T}_S(M_\mathrm{in}) := C\mathcal{L}_{(S,A)}^{-1}(BM_\mathrm{in}) + DM_\mathrm{in}.
\end{equation}
\end{defn}

\begin{thm}\label{thm: Closed-loop moment adjusted by compensator moment}
Given a plant~\eqref{eq: LTI system with u and mu} and a signal generator~\eqref{eq: Signal generator for mu}, suppose that $\sigma(A) \cap \sigma(S) = \emptyset$, and let $M_\mathrm{open}\in\mathbb{R}^{p\times\nu}$ be the open-loop moment of the system \eqref{eq: LTI system with u and mu} at $(S,L)$.
Let $\mathcal{T}_S(\cdot)$ be a moment transfer operator of $(A,B,C,D)$ at $S$.
For given $M_\mathrm{des}\in\mathbb{R}^{{p}\times\nu}$ and a compensator~\eqref{eq: compensator} such that $\sigma(A_\mathrm{cl}) \cap \sigma(S) = \emptyset$, we have
$M_\mathrm{cl}(F,G,H,\rho) = M_\mathrm{des}$, if and only if
the closed-loop moment of the compensator~\eqref{eq: compensator} is $M_\mathrm{c}$ which satisfies
\begin{equation}\label{eq: proper M_c}
\mathcal{T}_S(M_\mathrm{c}) = M_\mathrm{des} - M_\mathrm{open}.
\end{equation}
\end{thm}
\begin{proof}
\textbf{Sufficiency}:
The condition $\sigma(A_\mathrm{cl}) \cap \sigma(S) = \emptyset$ guarantees the existence and uniqueness of the solution $\Pi_x$ and $\Pi_\xi$ to \eqref{eq: closed sylvester - a,b}.
Since $\sigma(A) \cap \sigma(S) = \emptyset$, $\Pi_x = \mathcal{L}_{(S,A)}^{-1} (PL) + \mathcal{L}_{(S,A)}^{-1} (BH\Pi_\xi)$.
By applying \eqref{eq: M_cl expansion}, the closed-loop moment becomes $M_\mathrm{cl} = M_\mathrm{open} + \mathcal{T}_S(M_\mathrm{c})$. Due to the assumption \eqref{eq: proper M_c}, we obtain $M_\mathrm{cl} = M_\mathrm{open} + (M_\mathrm{des} - M_\mathrm{open}) = M_\mathrm{des}$.

\textbf{Necessity}:
Suppose that $M_\mathrm{cl} = M_\mathrm{des}$.
By \eqref{eq: M_cl expansion}, $M_\mathrm{cl} = M_\mathrm{open} +\mathcal{T}_S(M_\mathrm{c})$.
\end{proof}

As will be shown in Theorem~\ref{thm: canonical compensator}, provided that $M_\mathrm{c}$ satisfies \eqref{eq: proper M_c}, one can always construct a compensator whose closed-loop moment is $M_\mathrm{c}$. Consequently, the desired moment assignability depends entirely on the existence of a solution $M_\mathrm{c}$ to \eqref{eq: proper M_c} for a given desired moment difference $\Delta M := M_\mathrm{des} - M_\mathrm{open}$.
Therefore, the desired moment assignability is equivalent to the algebraic condition $\Delta M \in \mathcal{R}(\mathcal{T}_S)$.

\begin{rem}\label{rem: comparison with OR1}
Under the assumption $\sigma(A) \cap \sigma(S) = \emptyset$, regarding $y$ as a signal to be regulated, the existence of $M_\mathrm{c}$ satisfying \eqref{eq: proper M_c}, with $M_\mathrm{des}=0$ and $L = I_\nu$, is equivalent to the solvability of the regulator equations of output regulation. Consider the existence problem of $M_\mathrm{c}$ satisfying $\mathcal{T}_S(M_\mathrm{c}) = -M_\mathrm{open}$, i.e.,  $C\mathcal{L}_{(S,A)}^{-1}(BM_\mathrm{c}) + DM_\mathrm{c} = -(C\mathcal{L}_{(S,A)}^{-1}(P) + Q)$.
This is equivalent to the existence problem of $M_\mathrm{c}$ satisfying $C\Pi + DM_\mathrm{c} + Q = 0$ where $\Pi$ satisfies $\Pi S = A\Pi + BM_\mathrm{c} + P$.
Referring to \cite{huang2004nonlinear}, it is called regulator equations, which is crucial for the solvability of output regulation.
In conclusion, this structural decomposition of $\Pi$ opens the algebraic black box of the classical regulator equations, revealing the physical mechanism behind output regulation, which only utilizes the Sylvester equation.
\end{rem}

\begin{rem}\label{rem: comparison with OR2}
Even in the scenario where the desired moment $M_\mathrm{des}$ falls outside the range of the moment transfer operator $\mathcal{T}_S$, the proposed framework suggests an alternative as an optimization problem.
Instead of seeking an exact solution to \eqref{eq: proper M_c}, one can seek for the optimizing compensator moment $M_c^*$ that minimizes the moment difference 
$\|\Delta M - \mathcal{T}_S(M_\mathrm{c})\|$ then, the redefined optimal desired moment would become $M_\mathrm{des}^* := M_\mathrm{open} + \mathcal{T}_S(M_\mathrm{c}^*)$.
Note that the metric $\|\cdot\|$ can be freely chosen, allowing the designer to penalize certain disturbance or reference modes more heavily.
\end{rem}

\subsection{Details on the moment transfer operator}
Analyzing the moment transfer operator $\mathcal{T}_S$ is achieved by connecting it to the classical frequency-domain notion of moment.
\begin{defn}[\cite{antoulas2005approximation}]\label{def: classical moment}
Let $W(s) = C(sI-A)^{-1}B + D$ be a transfer function of LTI system $(A,B,C,D)$ and $s^\star\in\mathbb{C}\setminus\sigma(A)$.
The $0$-moment of system $(A,B,C,D)$ at $s^\star$ is the complex matrix $\eta_0(s^\star) = W(s^\star)$.
For $k\in\mathbb{Z}_{\ge1}$, the $k$-moment of system $(A,B,C,D)$ at $s^\star$ is the complex matrix
\begin{equation}
\eta_k(s^\star) = \frac{(-1)^k}{k!}\left[ \frac{d^k}{ds^k} W(s) \right]_{s=s^\star}.
\end{equation}
\end{defn}

\begin{thm}\label{thm: explicit moment transfer operator}
Consider the LTI system $(A,B,C,D)$, the real matrix $S\in\mathbb{R}^{\nu\times\nu}$ with $\sigma(A) \cap \sigma(S) = \emptyset$, and its moment transfer operator at $S$, $\mathcal{T}_S$.
Let $\Sigma \in \mathbb{C}^{\nu\times\nu}$ be the Jordan canonical form of $S$ whose explicit form is
\begin{equation}
\begin{bmatrix}
\Sigma_1 & ~ & ~ & ~ \\
~ & \Sigma_2 & ~ & ~ \\
~ & ~ & \ddots & ~ \\
~ & ~ & ~ & \Sigma_d
\end{bmatrix}
\end{equation}
where $d$ is the total number of Jordan blocks, and each block $\Sigma_i \in \mathbb{C}^{\nu_i \times \nu_i}$ is given by
\begin{equation}
\Sigma_i =
\begin{bmatrix}
s_i & 1 & ~ & ~ & ~ \\
~ & s_i & 1 & ~ & ~ \\
~ & ~ & \ddots & \ddots & ~ \\
~ & ~ & ~ & s_i & 1 \\
~ & ~ & ~ & ~ & s_i
\end{bmatrix}.
\end{equation}
Then, for all $M \in \mathbb{R}^{m \times \nu}$,
\begin{equation}\label{eq: explicit form of the moment transfer operator}
\mathcal{T}_S(M) = \sum_{i=1}^{d} \sum_{k=0}^{\nu_i - 1} \eta_k(s_i) M P_i (s_i I_\nu - S)^k P_i
\end{equation}
where $\eta_k(s_i)$ is the $k$-moment of $(A,B,C,D)$ at $s_i$, and 
$P_i\in\mathbb{C}^{\nu\times\nu}$ is a projection matrix onto the generalized eigenspace corresponding to the $i$-th Jordan block $\Sigma_i$.
\end{thm}
\begin{proof}
Let $T\in\mathbb{C}^{\nu\times\nu}$ be a nonsingular matrix satisfying $S = T^{-1}\Sigma T$.
Define $\tilde M := MT^{-1}$.
Then the Sylvester equation $\Pi S = A\Pi + BM$ is equivalent to
\begin{equation}\label{eq: Sylvester equation for Sigma}
\tilde \Pi \Sigma = A\tilde\Pi + B\tilde M
\end{equation}
where $\tilde \Pi = \Pi T^{-1}$.
Decomposing the matrices as $\tilde \Pi = [\tilde \Pi_1 \ \cdots \ \tilde \Pi_d]$ and $\tilde M = [\tilde M_1 \ \cdots \ \tilde M_d]$ where $\tilde\Pi_i \in \mathbb{C}^{n\times\nu_i}$ and $\tilde M_i \in \mathbb{C}^{m\times\nu_i}$,
\eqref{eq: Sylvester equation for Sigma} can be decomposed into $d$ equations
\begin{equation}
\tilde \Pi_i \Sigma_i = A\tilde\Pi_i + B\tilde M_i.
\end{equation}
According to \cite[Lemma 3]{shakib2023time}, we have
\begin{equation}
\left(C\tilde\Pi_i + D\tilde M_i \right)_j = \sum_{k=0}^{j-1}(-1)^k\eta_k(s_i)m_{i(j-k)}
\end{equation}
where $(C\tilde\Pi_i + D\tilde M_i)_j$ denotes the $j$-th column of $C\tilde\Pi_i + D \tilde M_i$ and $m_{ij}$ is the $j$-th column of $\tilde M_i$.
Therefore, we obtain that
\begin{equation}
\begin{aligned}
C\tilde \Pi_i + D\tilde M_i &= \eta_0(s_i)[m_{i1} \ m_{i2} \ \cdots \ m_{i\nu_i}]\\
& \quad - \eta_1(s_i) [0 \ m_{i1} \ \cdots \ m_{i(\nu_i-1)}] + \cdots \\
&\quad +(-1)^{\nu_i-1} \eta_{\nu_i-1}(s_i)[0 \ \cdots \ m_{i1}] \\
&= \sum_{k=0}^{\nu_i - 1} \eta_k(s_i) \tilde M_i (s_i I_{\nu_i} - \Sigma_i)^k.
\end{aligned}
\end{equation}
This procedure also yields
\begin{equation}
C\tilde\Pi + D\tilde M = \sum_{i=1}^{d} \sum_{k=0}^{\nu_i - 1} \eta_k(s_i) \tilde M E_i (s_i I_\nu - \Sigma)^k E_i
\end{equation}
where $E_i = \operatorname{blkdiag}(0_{\nu_1},\dots,0_{\nu_{i-1}},I_{\nu_i},0_{\nu_{i+1}},\dots,0_{\nu_d})$.
Using $\tilde \Pi = \Pi T^{-1}$, $\Sigma =  TST^{-1}$, and $\tilde M = M T^{-1}$, with $P_i := T^{-1}E_i T
$, one can obtain \eqref{eq: explicit form of the moment transfer operator}.
\end{proof}

\begin{cor}\label{cor: explicit T_S as a matrix}
The moment transfer operator of $(A,B,C,D)$ at $S$ has the matrix representation
\begin{equation}
\operatorname{vec}(\mathcal{T}_S(M)) = \mathbf{T}_S \operatorname{vec}(M)
\end{equation}
where
\begin{equation}
\mathbf{T}_S = (T\otimes I_m)^\top \mathbf{H} (T\otimes I_p)^{-\top}
\end{equation}
with matrices $E_i$ and $T$ in the proof of Theorem~\ref{thm: explicit moment transfer operator} and
\begin{equation} \label{eq: The H of moment transfer operator}
\mathbf{H} = \sum_{i=1}^{d} \sum_{k=0}^{\nu_i - 1} \left( E_i^\top (s_i I_\nu - \Sigma^\top)^k E_i^\top \right ) \otimes (\eta_k(s_i)).
\end{equation}
\end{cor}
\begin{proof}
It is straightforward by applying the following properties of the Kronecker product~\cite{graham2018kronecker} to Theorem~\ref{thm: explicit moment transfer operator} (assuming matrices have compatible dimensions and full rank where necessary):
\begin{itemize}
\item $\operatorname{vec}(AXB) = (B^\top \otimes A)\operatorname{vec}(X)$;
\item $(A \otimes B)(C \otimes D) = (AC) \otimes (BD)$;
\item $(A \otimes B)^\top = A^\top \otimes B^\top$;
\item $(A \otimes B)^{-1} = A^{-1} \otimes B^{-1}$.
\end{itemize}
\end{proof}

Corollary~\ref{cor: explicit T_S as a matrix} carries practical and theoretical implications. It translates the necessary and sufficient condition for moment assignability into a computationally tractable matrix condition $\operatorname{vec}(\Delta M) \in \mathcal{R}(\mathbf{T}_S)$.

Observing the structure of $\mathbf{H}$ in \eqref{eq: The H of moment transfer operator}, it possesses an explicit block-diagonal form $\mathbf{H} = \operatorname{blkdiag}(H_d, \dots, H_1) \in \mathbb{R}^{{p}\nu\times {m}\nu}$, where each block $H_i$ is a lower block-triangular matrix given by
\begin{equation}
H_i = 
\begin{bmatrix}
\eta_0(s_i) & ~ & ~ & ~ \\
-\eta_1(s_i) & \ddots & ~ & ~ \\
\vdots & \ddots & \eta_0(s_i) & ~ \\
(-1)^{\nu_i-1}\eta_{\nu_i - 1}(s_i) & \cdots & -\eta_1(s_i) &  \eta_0(s_i)
\end{bmatrix}.
\end{equation}
Due to this block-triangular structure, the rank of $\mathbf{H}$ is determined solely by its diagonal blocks, i.e., the $0$-moments $\eta_0(s_i)$.
It is known that if an eigenvalue $s_i$ of the signal generator coincides with a transmission zero of the plant, the matrix $\eta_0(s_i)$ becomes rank-deficient.
Consequently, this rank deficiency propagates to $\mathbf{H}$ and $\mathbf{T}_S$, thereby reducing the dimension of the assignable moment space.
In conclusion, under the condition that the plant is right-invertible at $\sigma(S)$ ,i.e., ${m \ge p}$ and $s_i\in\sigma(S)$ are not transmission zeros, the moment transfer operator is surjective. This guarantees that any desired closed-loop moment $M_\mathrm{des}$ can be assigned by a suitable choice of the compensator.
We do emphasize, however, even when some $s_i$’s are transmission zeros, $\Delta M \in \mathcal{R}(\mathcal{T}_S)$ can still be satisfied, which is simply enough for the moment assignment.

\section{Compensator for steady-state assignment}

\subsection{Canonical structure for moment assignment}

In this section, we introduce a canonical compensator structure that describes all the compensators assigning the closed-loop moment to $M_\mathrm{des}$.

According to the key observation made at the end of Section III-A, the necessary and sufficient condition for $M_\mathrm{des}$ to be assignable has been characterized in Section III-B as there exists a virtual signal generator $(S, M_\mathrm{c})$ yielding the open-loop moment at $(A, B, C, D)$ to be $M_\mathrm{des}$.
If this condition is met, on the other hand, the only job of the compensator, regardless of the specific choice of its internal parameters, is to transfer the moment $M_\mathrm{des}$ of the virtual signal generator $(S, M_\mathrm{des})$, corresponding to our system~\eqref{eq: LTI system with u and mu}, to the desired moment of the compensator $M_\mathrm{c}$, which then again acts as a virtual signal generator to the system~\eqref{eq: LTI system with u and mu}.

This perspective intuitively allows the full characterization of all the compensators assigning the closed-loop moment to $M_\mathrm{des}$, simply as all the system matching its open-loop moment to $M_\mathrm{c}$ at the virtual signal generator $(S, M_\mathrm{des})$.
Acknowledging this observation allows the use of conventional canonical structure for moment matching~\cite{astolfi2010model}, while we extend this with an additional freedom of choice as 
\begin{equation}\label{eq: compensator matching moment}
\begin{aligned}
\dot \xi_a &= (S - G_aM_\mathrm{des})\xi_a + F_{a} \xi_b + G_a u_\xi \\
\dot \xi_b &= -G_b M_\mathrm{des} \xi_a + F_{b} \xi_b + G_b u_\xi\\
y_\xi &= M_\mathrm{c} \xi_a + H_b\xi_b,
\end{aligned}
\end{equation}
parameterized by matrices $F_a, F_b, G_a, G_b$, and $H_b$, with the states $\xi_a \in \mathbb{R}^\nu$ and $\xi_b\in\mathbb{R}^{\rho - \nu}$, the input $u_\xi \in \mathbb{R}^{p}$, and the output $y_\xi \in \mathbb{R}^{m}$.

The architecture of the compensator consists of two parts to address dual objectives, moment assignment and stabilization.
The state $\xi_a$ constitutes a conventional canonical moment matching model~\cite{astolfi2010model} to match the moment $M_\mathrm{c}$ with respect to the virtual signal generator $(S, M_\mathrm{des})$.
On the other hand, $\xi_b$ represents redundant state variables for closed-loop stabilization while not altering the assigned moment.

The following theorem shows that, given $M_\mathrm{des}\in\mathbb{R}^{{p}\times\nu}$ and any proper $M_\mathrm{c}\in\mathbb{R}^{{m}\times \nu}$ achieving \eqref{eq: proper M_c}, the closed-loop moment of the proposed compensator~\eqref{thm: canonical compensator} is precisely $M_\mathrm{c}$, and hence the moment of the closed-loop system is assigned to $M_\mathrm{des}$.
Moreover, the theorem demonstrates that the proposed architecture parameterizes the entire family of compensators that assign the closed-loop moment to $M_\mathrm{des}$, up to pole-zero cancellations.

\begin{thm}\label{thm: canonical compensator}
Given matrix $M_\mathrm{des} \in \mathbb{R}^{{p}\times \nu}$ and the signal generator~\eqref{eq: Signal generator for mu}, consider the compensator~\eqref{eq: compensator matching moment} whose $M_\mathrm{c} \in \mathbb{R}^{{m}\times \nu}$ satisfies \eqref{eq: proper M_c}, assuming $\sigma(A_\mathrm{cl}) \cap \sigma (S) = \emptyset$.
Then, the closed-loop moment of the system~\eqref{eq: LTI system with u and mu} with the compensator~\eqref{eq: compensator matching moment} at $(S,L)$ is $M_\mathrm{des}$.
Furthermore, let $(\tilde F, \tilde G, \tilde H)$ be any minimally realized model of the order $\bar\rho$ such that its corresponding closed-loop matrix $\tilde{A}_\mathrm{cl}$ satisfies $\sigma(\tilde{A}_\mathrm{cl}) \cap \sigma(S) = \emptyset$ and it assigns the closed-loop moment of the system~\eqref{eq: LTI system with u and mu} to $M_\mathrm{des}$.
Then, there exists $\rho^*\in\mathbb{Z}_{\ge0}$ such that for all $\rho\ge\rho^*$, under the proper choice of $(F_a,F_b,G_a,G_b,H_b)$, $(\tilde F, \tilde G, \tilde H)$ is always algebraically equivalent to minimal realization of \eqref{eq: compensator matching moment}.
\end{thm}

\begin{proof}
The closed-loop moment of the compensator~\eqref{eq: compensator matching moment} is $M_\mathrm{c} \Pi_{\xi_a} + H_b \Pi_{\xi_b}$ where $\Pi_{\xi_a}\in\mathbb{R}^{\nu\times\nu}$ and $\Pi_{\xi_b}\in\mathbb{R}^{(\rho-\nu)\times\nu}$ satisfy
\begin{equation*}
\begin{bmatrix}
\Pi_x \\ \Pi_{\xi_a} \\ \Pi_{\xi_b}
\end{bmatrix} S
=
A_\mathrm{cl}
\begin{bmatrix}
\Pi_x \\ \Pi_{\xi_a} \\ \Pi_{\xi_b}
\end{bmatrix}  +  
\begin{bmatrix}
P \\ G_aQ \\ G_bQ
\end{bmatrix} L
\end{equation*}
where
\begin{equation*}
A_\mathrm{cl} =
\begin{bmatrix}
A & BM_C  & BH_b\\ G_aC & S - G_a(M_\mathrm{des}-DM_\mathrm{c}) & F_a + G_a D H_b \\ G_bC & -G_b (M_\mathrm{des} - DM_\mathrm{c}) & F_b + G_b D 
H_b
\end{bmatrix}.
\end{equation*}
By the spectrum condition, the above Sylvester equation has the unique solution and the solution is 
\begin{equation}\label{eq: solution of closed loop sylvester equation}
[(\mathcal{L}_{(S,A)}^{-1} (PL) + \mathcal{L}_{(S,A)}^{-1} (BM_\mathrm{c}))^\top \ \ I_\nu^\top \ \ 0_{(\rho-\nu)\times\nu}^\top]^\top.
\end{equation}
This could be easily shown by substituting the solution to the equation.
Therefore the closed-loop moment of the compensator~\eqref{eq: compensator matching moment} is $M_\mathrm{c}I_\nu + H_b0_{(\rho-\nu)\times\nu} = M_\mathrm{c}$.
Therefore, by Theorem~\ref{thm: Closed-loop moment adjusted by compensator moment}, the closed-loop moment of the system~\eqref{eq: LTI system with u and mu} with the compensator~\eqref{eq: compensator matching moment} at $(S,L)$ is $M_\mathrm{des}$.

Suppose minimally realized $\bar\rho$-order model
\begin{equation}\label{eq: any moment assigning compensator}
\begin{aligned}
\dot{\tilde{\xi}} &= \tilde F \tilde \xi + \tilde G u_\xi, & y_\xi = \tilde H \tilde \xi
\end{aligned}
\end{equation}
assigns closed-loop moment of the system~\eqref{eq: LTI system with u and mu} to $M_\mathrm{des}$.
According to Theorem~\ref{thm: Closed-loop moment adjusted by compensator moment} the closed-loop moment of the compensator~\eqref{eq: any moment assigning compensator} is $M_\mathrm{c}$.
By the definition, $M_\mathrm{c} = \tilde H \tilde \Pi_\xi$ where $\tilde\Pi_\xi \in \mathbb{R}^{\bar\rho\times\nu}$ satisfies
\begin{equation}\label{eq: Sylvester equation for compensator}
\begin{bmatrix}
\tilde \Pi_x \\ \tilde \Pi_\xi
\end{bmatrix}S
=
\begin{bmatrix}
A & B\tilde H \\ \tilde G C & \tilde F + \tilde G D\tilde H
\end{bmatrix}
\begin{bmatrix}
\tilde \Pi_x \\ \tilde \Pi_\xi
\end{bmatrix} + 
\begin{bmatrix}
P \\ \tilde G Q
\end{bmatrix}L.
\end{equation}
Note that $\tilde\Pi_x=\mathcal{L}_{(S,A)}^{-1}(PL + BM_\mathrm{c})$.
Let $\bar\nu := \operatorname{rank}(\tilde\Pi_\xi)$.
Then, there exist $T_\xi\in\mathbb{R}^{\bar\nu \times \bar\rho}$ and $V\in\mathbb{R}^{(\bar\rho-\bar\nu)\times \bar\rho}$ such that, for $\bar \Pi_\xi := T_\xi \tilde\Pi_\xi \in \mathbb{R}^{\bar\nu\times\nu}$, $\operatorname{rank}(\bar\Pi_\xi) = \bar \nu$ and $V\tilde\Pi_\xi = 0$.
Therefore, the matrix
\begin{equation}
T := \begin{bmatrix}
I_n & 0_{n \times \bar\rho} \\
0_{\bar\nu\times n} & T_\xi\\
0_{(\bar\rho - \bar\nu)\times n} & V
\end{bmatrix}
\end{equation}
is nonsingular.
Now, we define the coordinate transform $T [x^\top \ \tilde\xi^\top ]^\top =: [x^\top \ \bar \xi_a^\top \ \bar\xi_b^\top] ^\top\in \mathbb{R}^{n + \bar\nu + (\bar\rho-\bar\nu)}$, so that the dynamics is described by
\begin{equation}
\begin{aligned}
\dot x &= Ax + Bu + P\omega \\
y &= Cx + Du + Q\omega\\
\dot{\bar \xi}_a &=
\bar F_{aa}\bar \xi_a + \bar F_{ab}\bar\xi_b + \bar G_a y \\
\dot{\bar\xi}_b &=
\bar F_{ba}\bar \xi_a + \bar F_{bb}\bar\xi_b + \bar G_b y\\
u &= \bar H_a \bar{\xi}_a + \bar H_b \bar{\xi}_b.
\end{aligned}
\end{equation}
Then, \eqref{eq: Sylvester equation for compensator} is equivalent to
\begin{equation}\label{eq: Sylvester equation for any rho order model}
\begin{bmatrix}
\tilde\Pi_x \\ \bar\Pi_\xi \\ 0_{(\bar\rho-\bar\nu)\times\nu}
\end{bmatrix} S
=
\bar A_\mathrm{cl}
\begin{bmatrix}
\tilde\Pi_x \\ \bar\Pi_\xi \\ 0_{(\bar\rho-\bar\nu)\times\nu}
\end{bmatrix} +
\begin{bmatrix}
P \\ \bar G_a Q \\ \bar G_b Q
\end{bmatrix} L
\end{equation}
where
\begin{equation}
\bar A_\mathrm{cl}  =  
\begin{bmatrix}
A & B\bar H_a & B\bar H_b\\ \bar G_aC & \bar F_{aa} + \bar G_{a}D\bar H_{a} & \bar F_{ab} + G_{a}D\bar H_{b} \\ \bar G_bC & \bar F_{ba} + G_{b}D\bar H_{a} & \bar F_{bb} + G_{b}D\bar H_{b}
\end{bmatrix}.
\end{equation}
By the assumption, the closed-loop moment of the compensator $\bar H_a \bar\Pi_\xi + \bar H_b \cdot 0$ is $M_\mathrm{c}$, i.e., $\bar H_a \bar\Pi_\xi = M_\mathrm{c}$.
From the second and the last block rows, using $\tilde\Pi_x=\mathcal{L}_{(S,A)}^{-1}(PL + BM_\mathrm{c})$, one could obtain $\bar F_{aa} = \bar \Pi_\xi S \bar \Pi_\xi^\dagger - \bar G_a M_\mathrm{des}\bar\Pi_\xi^\dagger$ and $\bar F_{ba} = - \bar G_b M_\mathrm{des}\bar\Pi_\xi^\dagger$ where $\bar\Pi_\xi^\dagger$ denote the right inverse of $\bar\Pi_\xi$.
Therefore, $(\tilde F, \tilde G, \tilde H)$ is algebraically equivalent to
\begin{equation}\label{eq: Minimal realized compensator - transformed}
\begin{aligned}
\dot{\bar\xi}_a &= (\bar \Pi_\xi S \bar \Pi_\xi^\dagger - \bar G_a M_\mathrm{des}\bar\Pi_\xi^\dagger) \bar\xi_a + \bar F_{ab}\bar\xi_b + \bar G_a u_\xi \\
\dot{\bar\xi}_b &= - \bar G_b M_\mathrm{des}\bar\Pi_\xi^\dagger \bar\xi_a + \bar F_{ba}\bar\xi_b + \bar G_b u_\xi \\
y_\xi &= M_\mathrm{c}\bar \Pi_\xi^\dagger \bar\xi_a + \bar H_b\bar\xi_b,
\end{aligned} 
\end{equation}
which is minimal.
Finally, with any $\rho \ge (\bar \rho - \bar\nu) + \nu =:\rho^*$, letting
\begin{equation}
\begin{aligned}
F_a &= \bar \Pi_\xi^\dagger [\bar F_{ab} \ \ 0_{\bar\nu\times r}], & G_a &= \bar\Pi_\xi^\dagger \bar G_a, \\
F_b &= \begin{bmatrix}
\bar F_{bb} & 0 \\ 0 & F_{\mathrm{null}}
\end{bmatrix},
& G_b &= \begin{bmatrix}
\bar G_b \\ 0_{r\times m}
\end{bmatrix},\\
H_b &= [\bar H_b \ \ 0_{p\times r}],
\end{aligned}
\end{equation}
with $r = \rho - \rho^*\ge0$ and any real matrix $F_{\mathrm{null}} \in \mathbb{R}^{r \times r}$,
\eqref{eq: Minimal realized compensator - transformed} is the minimal realization of \eqref{eq: compensator matching moment}.
This completes the proof.
\end{proof}

Therefore, one can represent any moment assigning compensator \eqref{eq: compensator}, utilizing \eqref{eq: compensator matching moment} with sufficiently large $\rho$.

\subsection{Stabilizability}

Due to the discussion made in the previous section, let us confine ourselves now to the compensator of the form \eqref{eq: compensator matching moment} without loss of generality.
Then, the problem of whether the closed-loop stability can be achieved or not is equivalent to whether the matrix
\begin{equation}
A_\mathrm{cl} =
\begin{bmatrix}
A & BM_\mathrm{c}  & BH_b\\ G_aC & S - G_a(M_\mathrm{des}-DM_\mathrm{c}) & F_a + G_a D H_b \\ G_bC & -G_b (M_\mathrm{des} - DM_\mathrm{c}) & F_b + G_b D H_b
\end{bmatrix}
\end{equation}
can be made Hurwitz via appropriate choices of the free parameters $(F_a,F_b,G_a,G_b,H_b;\rho)$.

From the structure of the closed-loop matrix $A_\mathrm{cl}$, this problem can be equivalently formulated as the output feedback stabilization of the auxiliary augmented system
\begin{equation}\label{eq: Augmented system}{\small
\begin{aligned}
\begin{bmatrix}
\dot x \\ \dot \xi_a
\end{bmatrix} 
&= 
\underbrace{\begin{bmatrix}
A & BM_\mathrm{c} \\ G_a C & S-G_a(M_\mathrm{des} - DM_\mathrm{c})    
\end{bmatrix}}_{=: A_\mathrm{aug}}
\begin{bmatrix}
x \\ \xi_a
\end{bmatrix} + 
\underbrace{\begin{bmatrix}
B & 0 \\ G_a D & I_\nu
\end{bmatrix}}_{=:B_\mathrm{aug}} u \\
y &= \underbrace{
\begin{bmatrix}
C & -(M_\mathrm{des}-DM_\mathrm{c})
\end{bmatrix}
}_{=:C_\mathrm{aug}}
\begin{bmatrix}
x \\ \xi_a
\end{bmatrix}
\end{aligned}}
\end{equation}
with the output feedback controller
\begin{equation} \label{eq: output feedback compensator}
\dot \xi_b = (F_b + G_b D H_b)\xi_b + G_by, \quad u = 
\begin{bmatrix}
H_b \\ F_a
\end{bmatrix}\xi_b.
\end{equation}

Since \eqref{eq: output feedback compensator} can represent any LTI system (with an appropriate choice of $F_a$, $G_b$, $H_b$, and $\overline{F}_b := F_b + G_bDH_b$), it becomes straightforward that the problem of interest simplifies to whether the auxiliary system $(C_\mathrm{aug}, A_\mathrm{aug}, B_\mathrm{aug})$ can be made stabilizable and detectable via the choice of $G_a$.
The following lemmas dictate that, in fact, this is a systematic property that is independent with the choice of $G_a$.

\begin{lem}\label{lem: N&S condition ob stbl}
The pair $(A_\mathrm{aug},B_\mathrm{aug})$ is stabilizable, if and only if the pair $(A,B)$ is stabilizable.
\end{lem}

\begin{proof}
By the Popov-Belevich-Hautus test (PBH test)~\cite{trentelman2002control}, the pair $(A_\mathrm{aug},B_\mathrm{aug})$ is stabilizable, if and only if, for all $\lambda \in \mathbb{C}_{\ge0}$, the matirx
\begin{equation}
\begin{bmatrix}
A - \lambda I & B M_c & B &0 \\ G_a C & S-G_a(M_\mathrm{des}-DM_\mathrm{c}) - \lambda I & G_a D & I
\end{bmatrix}
\end{equation}
is full row rank.
Since performing elementary column operations does not adjust the rank, this is equivalent to
\begin{equation}
\operatorname{rank}
\begin{pmatrix}
A - \lambda I & 0 & B &0 \\ 0 & 0 & 0 & I\end{pmatrix} = n+\nu,
\end{equation}
for all $\lambda \in \mathbb{C}_{\ge0}$.
The equivalent condition for the full row rank is the stabilizability of the pair $(A,B)$.
\end{proof}

\begin{lem}\label{lem: N&S condition for detl}
Suppose $\sigma(A) \cap \sigma(S) = \emptyset$.
The pair $(C_\mathrm{aug},A_\mathrm{aug})$ is detectable,
if and only if, the pairs $(C,A)$ and $(M_\mathrm{open},S)$ are detectable.
\end{lem}

\begin{proof}
The pair $(C_\mathrm{aug},A_\mathrm{aug})$ is not detectable, if and only if, by the PBH test, for some $\lambda \in \mathbb{C}_{\ge0}$, there is a nonzero vector $\phi := [x_0^\top \ \omega_0^\top]^\top \in\mathbb{R}^{n+\nu}$ such that
\begin{equation}
\begin{bmatrix}
A - \lambda I & B M_\mathrm{c} \\ 0 & S -  \lambda I \\ C & -(M_\mathrm{des} - DM_\mathrm{c})
\end{bmatrix}
\begin{bmatrix}
x_0 \\ \omega_0
\end{bmatrix}
= 0
\end{equation}
where the row operation step is omitted for brevity but is straightforward.

If $\omega_0 = 0$, the block matrix equation reduces to $(A-\lambda I)x_0 = 0$ and $Cx_0 = 0$ for some $\lambda\in\mathbb{C}_{\ge0}$ and nonzero $x_0$.
This is the PBH condition for the pair $(C,A)$ being not detectable.

If $\omega_0 \neq 0$,
only investigating $\omega_0 \in \mathcal{N}(S-\lambda I)\setminus\{0\}$ where $\lambda\in\sigma(S)\cap\mathbb{C}_{\ge0}$ is sufficient to make $(S-\lambda I)\omega_0 = 0$.
From the first block row, we obtain
\begin{equation}
(A-\lambda I)x_0 + BM_\mathrm{c} \omega_0 = 0.    
\end{equation}
Due to the assumption $\sigma(A)\cap\sigma(S)=\emptyset$, $x_{0}=(\lambda I-A)^{-1}BM_\mathrm{c}\omega_{0}$.
Substituting this $x_0$ into the third block row yields
\begin{equation}\label{eq: eqiv. condition 1 for detectability}
C(\lambda I-A)^{-1}BM_\mathrm{c}\omega_{0} - (M_\mathrm{des}-DM_\mathrm{c})\omega_{0} = 0.    
\end{equation}
To figure out the meaning of~\eqref{eq: eqiv. condition 1 for detectability}, consider $\hat{\Pi} := \mathcal{L}_{(S,A)}^{-1}(BM_\mathrm{c})$, which is the unique solution to the Sylvester equation $\hat{\Pi} S = A \hat{\Pi} + BM_\mathrm{c}$.
Postmultiplying both sides by $\omega_0$ yields $\lambda\hat{\Pi}\omega_0 = A\hat{\Pi}\omega_0 + BM_\mathrm{c}\omega_0$, which directly implies $\hat{\Pi}\omega_0 = (\lambda I - A)^{-1}BM_\mathrm{c}\omega_0$.
Consequently, \eqref{eq: eqiv. condition 1 for detectability} is equivalent to
\begin{equation}
(C\hat{\Pi} + DM_\mathrm{c} - M_\mathrm{des})\omega_0 = 0.
\end{equation}
Recall that the assigned compensator moment $M_\mathrm{c}$ satisfies $M_\mathrm{des} - M_\mathrm{open} = C\hat{\Pi} + DM_\mathrm{c}$.
Substituting this into the above equation simplifies it entirely to $M_\mathrm{open}\omega_0 = 0$.
Combined with $(S-\lambda I)\omega_0 = 0$ for $\lambda \in \mathbb{C}_{\ge0}$ and $\omega_0 \neq 0$, this implies that the pair $(M_\mathrm{open}, S)$ is not detectable, by PBH test.
In conclusion, taking the contrapositive of both cases, the pair $(C_\mathrm{aug}, A_\mathrm{aug})$ is detectable if and only if both $(C,A)$ and $(M_\mathrm{open}, S)$ are detectable.
\end{proof}

The necessary and sufficient condition in Lemma~\ref{lem: N&S condition for detl} provides a physical insight into the detectability of the augmented system.
Mathematically, it implies that no unstable or oscillatory mode of the signal generator can reside entirely within the null space of the open-loop moment $M_\mathrm{open}$.
If such a shared null space exists, the open-loop plant would completely block the transmission of this specific mode, leaving no trace in the steady-state output.
Because the proposed dynamic compensator solely relies on the measured output $y$ to construct a stabilizing feedback, it would be inherently blind to this internal excitation.
Consequently, the augmented system loses detectability, making output-feedback stabilization impossible.
Furthermore, in the context of output regulation where $M_\mathrm{des} = 0$, this condition elegantly captures the classical limitation that stabilization fails when the disturbance frequency coincides with a transmission zero of the plant.

\begin{thm}\label{thm: Stabilizabilty of A_cl with assignment}
Suppose that $\sigma(A) \cap \sigma (S) = \emptyset$.
There exists a compensator of the form~\eqref{eq: compensator matching moment} such that stabilizes the closed-loop system, if and only if
$(A,B)$ is stabilizable, and both $(C,A)$ and $(M_\mathrm{open},S)$ are detectable.
\end{thm}

The proof follows directly from standard linear system theory, hence it is omitted.
By the separation principle, the closed-loop system matrix $A_\mathrm{cl}$ formulated by the canonical compensator structure \eqref{eq: compensator matching moment} can be made Hurwitz, if and only if the augmented pairs $(A_\mathrm{aug},B_\mathrm{aug})$ and $(C_\mathrm{aug},A_\mathrm{aug})$ associated with \eqref{eq: Augmented system} are stabilizable and detectable, respectively.
This establishes the exact necessary and sufficient condition for the existence of the stabilizing and moment-assigning compensator.
The explicit computation of the stabilizing gains is detailed in the subsequent constructive synthesis procedure.

Provided that $\sigma(A) \cap \sigma(S) = \emptyset$, the pair $(A,B)$ is stabilizable, and both $(C,A)$ and $(M_\mathrm{open}, S)$ are detectable, the proposed compensator \eqref{eq: compensator matching moment} can be systematically synthesized using the following procedure:
\begin{enumerate}
\item \textbf{Moment Assignment:} Compute the required compensator moment matrix $M_\mathrm{c}$ by solving the linear algebraic equation $\mathcal{T}_S(M_\mathrm{c}) = M_\mathrm{des} - M_\mathrm{open}$.
\item \textbf{Closed-loop Stabilization:} Choose any $G_a$. Then, determine the remaining parameters $(F_a,F_b,G_b,H_b;\rho)$ by designing a standard output feedback controller \eqref{eq: output feedback compensator} which stabilizes the augmented system \eqref{eq: Augmented system}.
\end{enumerate}

\section{Numerical Example}

We consider the example of tracking control of the NASA HiMAT aircraft model~\cite{safonov2003feedback} of the form of~\eqref{eq: LTI system with u and mu} with the disturbances driven by the signal generator of the form of~\eqref{eq: Signal generator for mu} where
\begin{equation}{\small
\begin{aligned}
&A = \begin{bmatrix} -0.0226 & -36.6 & -18.9 & -32.1 & 3.25 & -0.76 \\ 9.3\text{e-}5 & -1.90 & 0.983 & -7.3\text{e-}4 & -0.17 & -0.005 \\ 0.0123 & 11.7 & -2.63 & 8.8\text{e-}4 & -31.6 & 22.4 \\ 0 & 0 & 1 & 0 & 0 & 0 \\ 0 & 0 & 0 & 0 & -30 & 0 \\ 0 & 0 & 0 & 0 & 0 & -30 \end{bmatrix}, \\
&B = \begin{bmatrix} 0 & 0 & 0 & 0 & 30 & 0 \\ 0 & 0 & 0 & 0 & 0 & 30 \end{bmatrix}^\top, \quad C = \begin{bmatrix} 0 & 1 & 0 & 0 & 0 & 0 \\ 0 & 0 & 0 & 1 & 0 & 0 \end{bmatrix},\\
&S = \begin{bmatrix} 0 & 0 & 0 \\ 0 & 0 & 3 \\ 0 & -3 & 0 \end{bmatrix},\quad P = \begin{bmatrix} 0 & 1 & 1 & 0 & 0 & 0 \\ 0 & 1 & 0 & 0 & 0 & 0 \\ 0 & 0 & 1 & 0 & 0 & 0 \end{bmatrix}^\top, \\
&L = I_3,\quad Q = 0.\\
\end{aligned}}
\end{equation}
The signal generator $S$ is configured to represent a mixture of a constant wind and an oscillatory gust with a frequency of 3 rad/s, and the disturbance matrix $P$ is set to reflect a vertical gust affecting the angle of attack and the pitching moment.
The control objective is to perfectly reject the constant disturbance while attenuating the oscillatory disturbance, rather than eliminating it completely, therefore we let 
\begin{equation}
M_\mathrm{des} = \begin{bmatrix} 0 & 0.1 & 0 \\ 0 & 0 & 0.1 \end{bmatrix}.
\end{equation}

We selected arbitrary $G_a$ and designed using LQG to determine the remaining parameters $F_a$, $F_b$, $G_b$, and $H_b$.

The simulation results can be found in Figure~\ref{fig: HiMAT}.
Despite of the disturbances, it is shown that proposed compensator eliminated the constant bias, and the oscillations converge precisely to the assigned trajectory.
These results validate that the proposed framework can successfully enforce an arbitrary steady-state behavior determined by the designer, during rejecting the disturbances.
\begin{figure}[t]
    \centering
    \includegraphics[width=\linewidth]{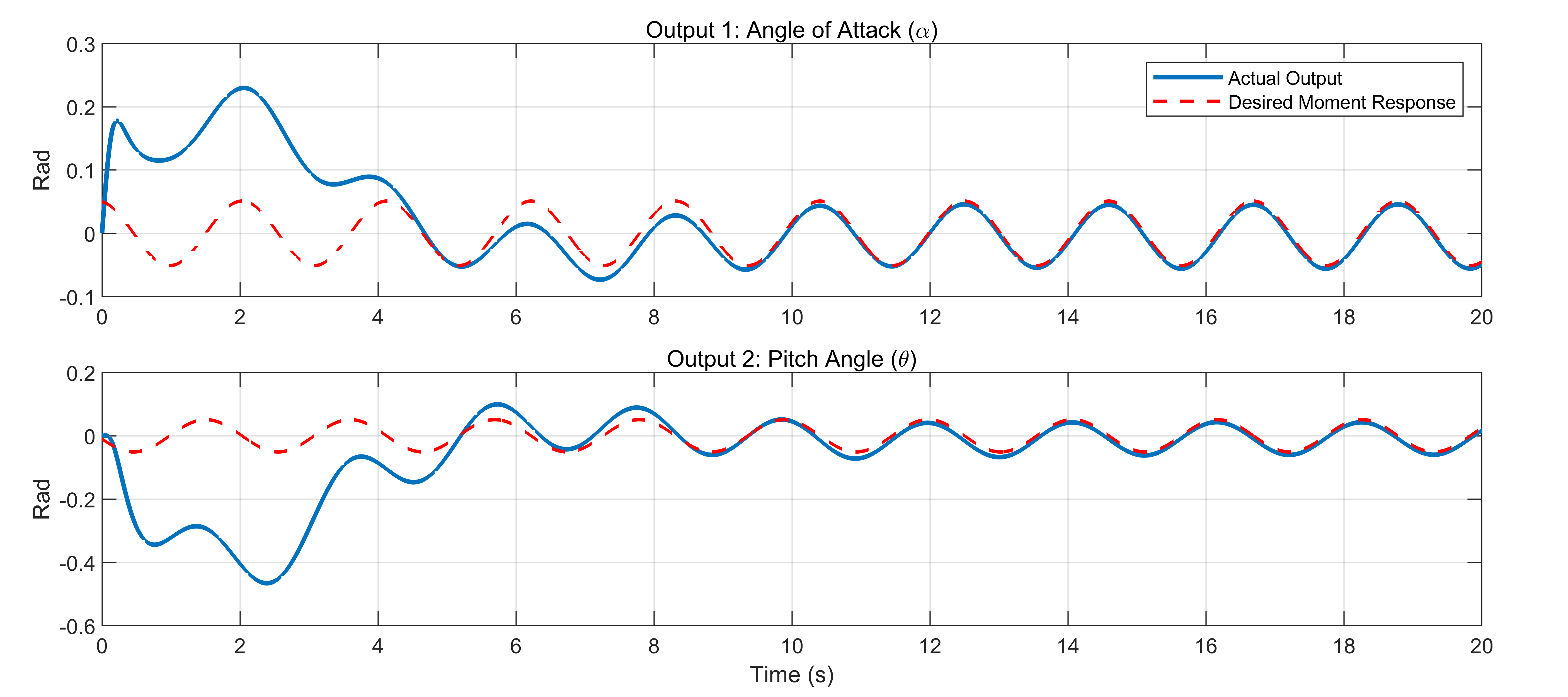}
    \caption{Disturbance rejection and tracking control of HiMAT aircraft model}
    \label{fig: HiMAT}
\end{figure}

\section{Conclusion}
In this paper, we have investigated the problem of moment assignment and stabilization, offering a novel perspective on the disturbance rejection and tracking problem through the lens of moment.
We established the equivalent conditions for moment assignability and formalized this concept by introducing the notion of the moment transfer operator.
Furthermore, we parameterized the universal structure of the compensator, effectively capturing the class of all possible compensators that assign the closed-loop moment to a desired target.
Finally, we derived the necessary and sufficient condition for the internal stabilizability of the closed-loop system and provided a systematic synthesis procedure to construct a stabilizing and moment assigning compensator.

The proposed framework conceptually unifies existing control paradigms.
It covers closed-loop interpolation by setting $M_\mathrm{des} = M_\mathrm{open}$ and classical output regulation by targeting a zero moment, i.e., $M_\mathrm{des} = 0$.
Furthermore, it provides an intuitive interpretation of output regulation.
By reformulating the regulator equations through the language of moment, their physical meaning is now evident: they merely represent the algebraic computation of a proper compensator moment $M_\mathrm{c}$ required to assign the desired moment $M_\mathrm{des}$.
By reducing the existence condition of a stabilizing compensator to the detectability of the pair $(M_\mathrm{open},S)$, we establish a physical intuition for stabilizability, making it straightforward to determine whether steady-state response assignment is feasible.
Parameterization of every controller achieving output regulation is also provided by this framework.

As moment can be extracted from input-output data without requiring full knowledge of the state-space matrices~\cite{scarciotti2017data}, this work naturally paves a road to data-driven output regulation.
Since the notion of moments has already been rigorously generalized beyond linear systems, this framework provides a foundation for a generalized approach to steady-state assignment in nonlinear, hybrid, and stochastic systems.

\bibliographystyle{IEEEtran}
\bibliography{IEEEabrv, ref}

\begin{thebibliography}{10}
\providecommand{\url}[1]{#1}
\csname url@samestyle\endcsname
\providecommand{\newblock}{\relax}
\providecommand{\bibinfo}[2]{#2}
\providecommand{\BIBentrySTDinterwordspacing}{\spaceskip=0pt\relax}
\providecommand{\BIBentryALTinterwordstretchfactor}{4}
\providecommand{\BIBentryALTinterwordspacing}{\spaceskip=\fontdimen2\font plus
\BIBentryALTinterwordstretchfactor\fontdimen3\font minus \fontdimen4\font\relax}
\providecommand{\BIBforeignlanguage}[2]{{%
\expandafter\ifx\csname l@#1\endcsname\relax
\typeout{** WARNING: IEEEtran.bst: No hyphenation pattern has been}%
\typeout{** loaded for the language `#1'. Using the pattern for}%
\typeout{** the default language instead.}%
\else
\language=\csname l@#1\endcsname
\fi
#2}}
\providecommand{\BIBdecl}{\relax}
\BIBdecl

\bibitem{astolfi2010model}
A.~Astolfi, ``Model reduction by moment matching for linear and nonlinear systems,'' \emph{IEEE Trans. Autom. Control}, vol.~55, no.~10, pp. 2321--2336, 2010.

\bibitem{scarciotti2016model}
G.~Scarciotti and A.~Astolfi, ``Model reduction for hybrid systems with state-dependent jumps,'' \emph{IFAC-PapersOnLine}, vol.~49, no.~18, pp. 850--855, 2016.

\bibitem{scarciotti2021moment}
G.~Scarciotti and A.~R. Teel, ``On moment matching for stochastic systems,'' \emph{IEEE Trans. Autom. Control}, vol.~67, no.~2, pp. 541--556, 2021.

\bibitem{moreschini2024closed}
A.~Moreschini and A.~Astolfi, ``Closed-loop interpolation by moment matching for linear and nonlinear systems,'' \emph{IEEE Trans. Autom. Control}, vol.~70, no.~5, pp. 2918--2933, 2024.

\bibitem{bellman1997introduction}
R.~Bellman, \emph{Introduction to matrix analysis}.\hskip 1em plus 0.5em minus 0.4em\relax Philadelphia, PA, USA: SIAM, 1997.

\bibitem{huang2004nonlinear}
J.~Huang, \emph{Nonlinear output regulation: theory and applications}.\hskip 1em plus 0.5em minus 0.4em\relax Philadelphia, PA, USA: SIAM, 2004.

\bibitem{antoulas2005approximation}
A.~C. Antoulas, \emph{Approximation of large-scale dynamical systems}.\hskip 1em plus 0.5em minus 0.4em\relax Philadelphia, PA, USA: SIAM, 2005.

\bibitem{shakib2023time}
M.~F. Shakib, G.~Scarciotti, A.~Y. Pogromsky, A.~Pavlov, and N.~van~de Wouw, ``Time-domain moment matching for multiple-input multiple-output linear time-invariant models,'' \emph{Automatica}, vol. 152, no. 110935, 2023.

\bibitem{graham2018kronecker}
A.~Graham, \emph{Kronecker products and matrix calculus with applications}.\hskip 1em plus 0.5em minus 0.4em\relax New York, NY, USA: Dover, 2018.

\bibitem{trentelman2002control}
H.~L. Trentelman, A.~A. Stoorvogel, and M.~Hautus, \emph{Control theory for linear systems}.\hskip 1em plus 0.5em minus 0.4em\relax New York, NY, USA: Springer, 2001.

\bibitem{safonov2003feedback}
M.~Safonov, A.~Laub, and G.~Hartmann, ``Feedback properties of multivariable systems: The role and use of the return difference matrix,'' \emph{IEEE Trans. Autom. Control}, vol.~26, no.~1, pp. 47--65, 2003.

\bibitem{scarciotti2017data}
G.~Scarciotti and A.~Astolfi, ``Data-driven model reduction by moment matching for linear and nonlinear systems,'' \emph{Automatica}, vol.~79, pp. 340--351, 2017.

\end{thebibliography}

\end{document}